\newif\ifarxiv
\arxivtrue

\newcommand{\TITEL}{The field of the Reals and the Random Graph are
  not Finite-Word Ordinal-Automatic}

\newcommand{\PERSON}{
\author{Alexander Kartzow}
\address{Institut f\"ur Informatik, Universit\"at Leipzig, Germany\\
  Department f\"ur Elektrotechnik und Informatik, Universit\"at
  Siegen, Germany}
\email{kartzow@informatik.uni-siegen.de}
}

\newcommand{\ABSTRACT}{ Recently, Schlicht and Stephan lifted the
  notion of automatic-structures to the notion of (finite-word)
  ordinal-automatic structures. These are structures whose domain and
  relations can be represented by automata reading finite words whose
  shape is some fixed ordinal $\alpha$.  We lift \Delhomme's
  relative-growth-technique from the automatic and tree-automatic
  setting to the ordinal-automatic setting. This result implies that
  the random graph is not ordinal-automatic and infinite integral
  domains are not ordinal-automatic with respect to ordinals below
  $\OmegaPlus$ where $\omega_1$ is the first uncountable ordinal.}

\newcommand{\KEYWORDS}{Ordinal-automatic structures, automatic
  integral domains, Rado graph, growth rates}

\ifarxiv
\documentclass{lmcslike}

\else

\documentclass{article}
\usepackage{jalc}

\title{\TITEL
\ifarxiv \else \thanks{This work is supported by the  DFG
    research project GELO.} \fi}

\PERSON

\abstract{\ABSTRACT}

\keywords{\KEYWORDS} 

\newcommand{\qedhere}{}

\fi

\newtheorem{theorem}{Theorem}
\newtheorem{lemma}[theorem]{Lemma}
\newtheorem{corollary}[theorem]{Corollary}
\newtheorem{proposition}[theorem]{Proposition}

\ifarxiv
\theoremstyle{definition}\newtheorem{definition}[theorem]{Definition}
\theoremstyle{definition}\newtheorem{example}[theorem]{Example}
\theoremstyle{definition}\newtheorem{remark}[theorem]{Remark}

\else

\theorembodyfont{\rm}
\newtheorem{definition}[theorem]{Definition}

\newtheorem{remark}[theorem]{Remark}

\fi

\usepackage{amssymb}

\usepackage{mathtools}  

\usepackage{enumerate}
\usepackage{xspace}
\usepackage{verbatim} 

\usepackage{braket}  

\newcommand*{\automatic}[1]{\ensuremath{(#1)}\mbox{-automatic}}
\newcommand*{\OmegaPlus}{\ensuremath{{\omega_1+\omega^\omega}}}

\usepackage{xifthen} 

\newcommand{\N}{\mathbb{N}}

\newcommand{\Struc}[1]{\mathfrak{#1}}
\newcommand{\Aut}[1]{\mathcal{#1}}
\newcommand{\Fam}[1]{\mathcal{#1}}

\newcommand*{\supp}[0]{\mathsf{supp}} 

\newcommand*{\FinWords}[2]{{\ensuremath{{#1}^{(#2)}}}} 

\newcommand*{\run}[2]{\xrightarrow[{#2}]{#1}}

\newcommand*{\Buchi}{B\"uchi\xspace}
\newcommand*{\Todorcevic}{Todor\v{c}evi\'{c}\xspace} 

\renewcommand{\phi}{\varphi}

\newcommand{\Delhomme}{Delhomm\'e\xspace}

\usepackage{letltxmacro}  
\LetLtxMacro{\Standardexists}{\exists}    

\newcommand{\existswrapper}[1][]{\ifthenelse{\isempty{#1}}{}{#1\,}}

\renewcommand*{\exists}[1]{\Standardexists \existswrapper[#1]}
\LetLtxMacro{\Standardforall}{\forall}
\renewcommand*{\forall}[1]{\Standardforall \existswrapper[#1]}

\usepackage{hyperref}
\hypersetup{%
colorlinks=true,
pdfauthor={Alexander Kartzow},
pdftitle={\TITEL},
pdfkeywords={\KEYWORDS}
}

\begin{document}

\ifarxiv
\title{\TITEL}

\PERSON

\keywords{\KEYWORDS}
\subjclass[2012]{[{\bf Theory of computation}]: 
  Formal languages and automata theory---Automata over infinite objects}  
\titlecomment{This work is supported by the  DFG
    research project GELO}

\begin{abstract}
\ABSTRACT
\end{abstract}

\fi

\maketitle

\section{Introduction}

Finite automata  play a crucial
role in many areas of computer science. In particular, finite automata
have been used to represent certain infinite
structures. The basic notion of this branch of research is the class of
automatic structures (cf. \cite{KhoussainovN94}). A
structure is automatic if its domain as well 
as its relations are recognised by (synchronous multi-tape) finite
automata processing finite words. This class has the remarkable
property that the first-order theory of any automatic structure is
decidable. One goal in the theory of automatic structures is a
classification of those structures that are automatic
(cf.~\cite{Delhomme04,KhoussainovRS05,KhoussainovNRS07,DBLP:journals/apal/KhoussainovM09,KuLiLo11}). 
Besides finite automata reading finite or infinite (i.e.,
$\omega$-shaped) words
there are also finite automata reading finite or infinite
\emph{trees}. Using such automata as representation of 
structures leads to the notion of tree-automatic structures
\cite{Blumensath1999}. The
classification of tree-automatic structures is less advanced but some
results have been obtained in the last years 
(cf.~\cite{Delhomme04,Huschenbett13,KaLiLo12}).  
Schlicht and Stephan \cite{SchlichtS13} and Finkel and \Todorcevic
\cite{FinkelT11}
have started research on a new branch of automatic
structures based on automata processing $\alpha$-words where $\alpha$
is some ordinal. An $\alpha$-word is a map $w\in\Sigma^\alpha$ for some
finite alphabet $\Sigma$. We call $w$ a finite $\alpha$-word if there
is one symbol $\diamond$ such that $w(\beta)=\diamond$ for all but
finitely many ordinals $\beta < \alpha$. We call the  structures
represented by finite-word $\alpha$-automatic structures
\automatic{\alpha}. Many of the fundamental results on
automatic structures have analogues in the setting of
\automatic{\alpha} structures.
\begin{itemize}
\item The first-order theory of every \automatic{\alpha} structure is
  decidable and the class of \automatic{\alpha} structures is closed
  under expansions by first-order definable relations for all 
  $\alpha < \OmegaPlus$ 
  \cite{HuKaSchl14}.
\item Lifting a result of Blumensath \cite{Blumensath1999} from the
  word- and tree-automatic setting, 
  there is an \automatic{\alpha} structure which is complete
  for the class of \automatic{\alpha} structures under first-order
  interpretations \cite{HuKaSchl14}.
\item The sum-of-box-augmentation technique of \Delhomme
  \cite{Delhomme04} for tree-automatic structures has an analogue for
  ordinal-automatic structures 
  which allows to classify all \automatic{\alpha} ordinals 
  \cite{SchlichtS13} and give
  sharp bounds on the ranks of \automatic{\alpha} scattered linear
  orderings \cite{SchlichtS13}
  and well-founded order trees \cite{KartzowS13}. 
\item The word-automatic Boolean
  algebras \cite{KhoussainovNRS07} and  the \automatic{\omega^n}
  Boolean algebras \cite{HuKaSchl14}
  have been classified. In contrast, a
  classification of the tree-automatic Boolean algebras is still open.
\end{itemize}
In summary one can say that all known techniques which allow to prove
that a structure is not tree-automatic have known counterparts for
ordinal-automaticity. The only  exception to this rule has been \Delhomme's growth-rate-technique
\cite{Delhomme04}. We close this gap by showing that the maximal growth
rates of ordinal automatic structures also has a polynomial bound.  This
allows to show that the Rado graph is not
\automatic{\alpha}. In fact, we show that the bound on the maximal growth-rate of
\automatic{\alpha} structure that we provide is strictly smaller than
the bound for tree-automatic and strictly greater than the bound for
word-automatic structures. Exhibiting this fact, we provide a new
example of a structure that is \automatic{\omega^2} but not
word-automatic. This example also shows that our growth-rate bound for
\automatic{\alpha} structure is essentially optimal.

One of the long-standing open problems in the field of automatic
structures is the question whether the field of the reals 
$\Struc{R} = (\mathbb{R}, +, \cdot, 0, 1)$ has a presentation based on finite
automata. Due to cardinality reasons it is clear that this structure
is not word- or tree-automatic. Recently, 
Zaid et al.~\cite{ZaidGKP14} have shown that $\Struc{R}$ (as well as every 
infinite integral domain) is not infinite-word-automatic. This leaves
infinite-tree-automata as the last classical candidate that might
allow to represent $\Struc{R}$. Note that the cardinality argument
also shows that $\Struc{R}$ is not
\automatic{\alpha} for all countable $\alpha$ (because the set of finite
$\alpha$-words is countable). Nevertheless the set of finite
$\omega_1$-words is uncountable  whence $\Struc{R}$ may be a priori
\automatic{\alpha} for some uncountable ordinal $\alpha$. Using the
growth rate argument we can show that
no infinite
integral domain is \automatic{\alpha} for any ordinal $\alpha < \OmegaPlus$. 
Let us mention that it also remains open whether $\Struc{R}$ is
automatic with respect to automata that also accept infinite
$\alpha$-words for some $\alpha \geq \omega^2$.

\subsection{Outline of the Paper}
In the next section we recall the necessary definitions on
\automatic{\alpha} structures and the fundamental notions concerning
growth rates. In Section~\ref{sec:BasicResults} we recall basic results on
\automatic{\alpha} structures which are needed to obtain the growth
rate bound in Section~\ref{sec:GrowthRate}. Finally, Section~\ref{sec:Applications}
contains applications of the growth rate argument to the random graph,
integral domains and concludes with the construction of a new example
of an \automatic{\omega^2} structure which is not word-automatic
because its growth-rate exceeds the known bound for word-automatic
structures.

\section{Definitions}
\label{sec:defs}

\subsection{Ordinals}

We identify an ordinal $\alpha$ with the set of smaller
ordinals \mbox{$\{\beta \mid \beta < \alpha\}$}. We say $\alpha$ has
\emph{countable cofinality} if $\alpha=0$ or there is a  sequence
$(\alpha_i)_{i\in\omega}$ of ordinals such that 
$\alpha = \sup \{\alpha_i+1\mid i\in\omega\}$.  Otherwise we say $\alpha$ has \emph{uncountable cofinality}.
We denote the first uncountable ordinal by $\omega_1$. Note that it is
the first ordinal with uncountable cofinality.

For every ordinal $\alpha$ and every $n\in \N$, let $\alpha_{\sim n}$
be the ordinal of the form $\alpha_{\sim n} = \omega^{n+1}\beta$ for
some ordinal $\beta$ such that
\begin{equation*}
  \alpha = \alpha_{\sim n}+\omega^n m_n + \omega^{n-1} m_{n-1}+
  \dots + m_0  
\end{equation*}
for some natural numbers $m_0, \dots, m_n$.

\subsection{Ordinal-Shaped Words}

\emph{First of all, we agree on the following convention:} In this 
article, every alphabet $\Sigma$ contains a distinguished
\emph{blank symbol} which is denoted by $\diamond_\Sigma$ or, if the
alphabet is clear from the context, just by $\diamond$. Moreover, for
alphabets $\Sigma_1,\dotsc,\Sigma_r$, the distinguished symbol of the
alphabet $\Sigma_1 \times \dotsb \times \Sigma_r$ will always be
$\diamond_{\Sigma_1 \times \dotsb \times \Sigma_r} =
(\diamond_{\Sigma_1},\dotsc,\diamond_{\Sigma_r})$. 

For some limit ordinal $\beta\leq\alpha$ and a map
\mbox{$w:\alpha+1\to A$} we introduce the
following notation for the \emph{set of images cofinal in
$\beta$}:
\begin{equation*}
  \lim_{\beta} w:= \{a\in A \mid \forall{\beta'<\beta}
  \exists{\beta' < \beta'' < \beta}  w(\beta'') = a\}.
\end{equation*}

\begin{definition}
  An \emph{$(\alpha)$-word (over $\Sigma$)} (called a finite
  $\alpha$ word over $\Sigma$) is a map 
  $w\colon \alpha \to \Sigma$
  whose \emph{support}, i.e., the set
  \begin{equation*}
  	\supp(w) = \Set{\beta \in \alpha | w(\beta)\neq\diamond },
  \end{equation*}
  is finite.
  The set of all $(\alpha)$-words over $\Sigma$ is denoted by $\FinWords{\Sigma}{\alpha}$.
  We write $\diamond^\alpha$ for the constantly $\diamond$ valued word
  $w:\alpha\to \Sigma$, $w(\beta)=\diamond$ for all
  $\beta<\alpha$. 
\end{definition}

\begin{definition}
  If $\gamma \leq \delta \leq \alpha$ are ordinals 
  and $w:\alpha\to \Sigma$ some
  $(\alpha)$-word, we denote by
  $w{\restriction}_{[\gamma,\delta)}$ the restriction of $w$ to the
  subword between position $\gamma$ (included) and $\delta$
  (excluded). 
\end{definition}

\subsection{Automata and Automatic Structures}
\Buchi \cite{Buchi65} has already introduced automata
that process $(\alpha)$-words. These behave like usual finite automata
at successor ordinals while at limit ordinals a limit
transition that resembles the acceptance condition of a
Muller-automaton is used.

\begin{definition}
  An \emph{ordinal automaton} is a tuple $(Q, \Sigma, I, F, \delta)$ where
  $Q$ is a finite set of states, $\Sigma$ a finite alphabet, $I\subseteq Q$
  the initial states, $F\subseteq Q$ the final states and 
  \begin{equation*}
    \delta \subseteq  (Q\times \Sigma \times Q) \cup (2^Q\times Q)    
  \end{equation*}
  is the   transition relation. 
\end{definition}

\begin{definition}
  A \emph{run} of $\Aut{A}$ on the $(\alpha)$-word $w \in \FinWords{\Sigma}{\alpha}$ is a map 
  $r: \alpha+1 \to Q$ such that
  \begin{itemize}
  \item $ \left( r(\beta), w(\beta), r(\beta+1)\right)\in \Delta$ for
    all $\beta<\alpha$
  \item $(\lim_{\beta} r, r(\beta))\in\Delta$
    for all limit ordinals $\beta\leq\alpha$. 
  \end{itemize}
  The run $r$ is \emph{accepting} if 
  $r(0)\in I$ and
  $r(\alpha)\in F$.
  For $q,q' \in Q$, we write $q \run{w}{\Aut{A}} q'$ if there is  a run
  $r$ of $\Aut{A}$ on $w$
  with $r(0) = q$ and $r(\alpha) = q'$.
\end{definition}

In the following, we always fix an ordinal $\alpha$ and then
concentrate on the set of $(\alpha)$-words that a given ordinal automaton
accepts. In order to stress this fact, we will call the ordinal-automaton
an $(\alpha)$-automaton. 

\begin{definition}
  Let $\alpha$ be some ordinal and $\Aut{A}$ be an $(\alpha)$-automaton. 
  The \emph{$(\alpha)$-language} of $\Aut{A}$, denoted by $L_{\FinWords{\Sigma}{\alpha}}(\Aut{A})$,
  consists of all
  $(\alpha)$-words $w \in \FinWords{\Sigma}{\alpha}$ which admit
  an accepting run of $\Aut{A}$ on $w$. 
  Whenever $\alpha$ is clear from the context, we may omit the subscript
  $\FinWords{\Sigma}{\alpha}$ and just write $L(\Aut{A})$ instead of $L_{\FinWords{\Sigma}{\alpha}}(\Aut{A})$.
\end{definition}

Automata on words (or infinite words or (infinite) trees) have
been applied fruitfully for representing structures. This can be
lifted to the setting of $(\alpha)$-words and leads to the notion of
\automatic{\alpha} structures.
In order to use $(\alpha)$-automata to recognise relations of $(\alpha)$-words,
we need to encode tuples of $(\alpha)$-words by one $(\alpha)$-word:

\begin{definition}
  Let $\Sigma$ be an alphabet and $r \in \N$.
  \begin{enumerate}[(1)]
  \item We regard any tuple
    $\bar{w} = (w_1,\dotsc,w_r) \in \bigl(\FinWords{\Sigma}{\alpha}\bigr)^r$
    of $(\alpha)$-words over some alphabet $\Sigma$ as an $(\alpha)$-word
    $\bar{w} \in \FinWords{(\Sigma^r)}{\alpha}$ over the alphabet
    $\Sigma^r$ by defining
    \begin{equation*}
      \bar{w}(\beta) = \bigl(w_1(\beta),\dotsc,w_r(\beta)\bigr)
    \end{equation*}
    for each $\beta < \alpha$.
  \item An \emph{$r$-dimensional $(\alpha)$-automaton over $\Sigma$}
    is an $(\alpha)$-automaton $\Aut{A}$ over $\Sigma^r$. The $r$-ary relation on
    $\FinWords{\Sigma}{\alpha}$ \emph{recognised} by $\Aut{A}$ is denoted
    \begin{equation*}
      R(\Aut{A}) = \Set{\bar w \in \bigl(\FinWords{\Sigma}{\alpha}\bigr)^r | \bar w \in L(\Aut{A}) } \,.
    \end{equation*}
  \end{enumerate}
\end{definition}
  Usually, this interpretation of $\bar{w}$ as an $(\alpha)$-word is
  called \emph{convolution} of $\bar{w}$ and denoted $\otimes \bar{w}$.
  For the sake of convenience, we just omit the symbol $\otimes$.

\begin{definition} 
  Let $\tau = \{R_1, R_2, \dots, R_m\}$ be a finite relational signature
  and let relation symbol $R_i$ be of arity $r_i$. 
  A structure $\Struc{A}=(A,R^\Struc{A}_1,R^\Struc{A}_2, \dots, R^\Struc{A}_m)$ is
  \emph{\automatic{\alpha}} if there are an alphabet $\Sigma$ and
  $(\alpha)$-automata
  $\Aut{A}, \Aut{A}_\approx,\Aut{A}_1, \dots, \Aut{A}_m$ such that
  \begin{itemize}
  \item
    $\Aut{A}$ is an $(\alpha)$-automaton over $\Sigma$,
  \item 
    for each $R_i\in \tau$, $\Aut{A}_i$ is an $r_i$-dimensional $(\alpha)$-automaton over $\Sigma$
    recognising an $r_i$-ary relation $R(\Aut{A}_i)$ on $L(\Aut{A})$,
  \item $\Aut{A}_\approx$ is a $2$-dimensional $(\alpha)$-automaton over $\Sigma$
    recognising a congruence relation $R(\Aut{A}_\approx)$ on
    the structure
    $\Struc{A}'= \left (L(\Aut{A}), L(\Aut{A}_1), \dotsc, L(\Aut{A}_m)\right)$, and
  \item  the quotient structure $\Struc{A}'/R(\Aut{A}_\approx)$ is isomorphic to $\Struc{A}$,
    i.e.,  $\Struc{A}'/R(\Aut{A}_\approx)\cong \Struc{A}$.
  \end{itemize}
  In this situation, we call the tuple $(\Aut{A}, \Aut{A}_\approx,\Aut{A}_1, \dots, \Aut{A}_m)$ an
  \emph{\automatic{\alpha} presentation} of $\Struc{A}$. 
  This presentation is said to be \emph{injective} if $L(\Aut{A}_\approx)$
  is the identity relation on $L(\Aut{A})$. In this case, we usually omit
  $\Aut{A}_\approx$ from the tuple of automata forming the presentation.
\end{definition}

\subsection{Definitions Concerning Growth Rates}
The basic idea behind the growth rate technique is the question how
many elements of a structure can be distinguished using a fixed finite
set of relations and a set of parameters which has $n$ elements. We
call two elements $a$ and $b$ distinguishable by a $(1+p)$-ary relation $R$ with
parameters from $E$ if there are $e_1, e_2, \dots, e_p\in E$ such that
$(a, e_1, e_2, \dots, e_p)\in R$ while 
$(b, e_1, e_2, \dots, e_p)\notin R$. If $\lvert E \rvert = n$ and $R$
is some relation, it is clear that there are at most $2^{n^p}$ many
elements that are pairwise distinguishable by $E$ with parameters from
$E$. \Delhomme \cite{Delhomme04} has shown that for every tree-automatic
relation $R$ there are always sets $E$ with $n$ elements such that
there are at most $n^c$ pairwise distinguishable elements where $c$
is a constant only depending on $R$ (and not on $n$ or $E$). For
word-automatic structures this bound even drops to $n\cdot c$. We now
provide basic definitions that allow to derive a similar bound for
\automatic{\alpha} structures.

\begin{definition}
  Let $\Struc{A}$ be an \automatic{\alpha} structure with domain $A$
  and $\Phi$ be a finite set of $(\alpha)$-automata such that each
  $\Aut{A}\in\Phi$ recognises a $1+p$-ary relation $R_{\Aut{A}}
  \subseteq A^{1+p}$. Let $E\subseteq A$ be a finite set and let
  $\Fam{F}$ be an infinite family of subsets of $A$ with $\emptyset\in\Fam{F}$. 
  \begin{enumerate}
  \item 
    For all
    $a,a'\in A$ we write $a\sim^\Phi_E a'$ if 
    \begin{equation*}
      (a, e_1, \dots, e_p) \in R_{\Aut{A}} \Leftrightarrow (a', e_1, \dots, e_p) \in R_{\Aut{A}}    
    \end{equation*}
    for all $e_1, \dots, e_p \in E$ and all $\Aut{A}\in \Phi$, i.e., $a$
    and $a'$ are indistinguishable with the automata from $\Phi$ and
    parameters in $E$.
  \item 
    We say $S\subseteq A$ is \emph{$E$-$\Phi$-free} if $a
    \not\sim^\Phi_E a'$ for all $a, a'\in S$.
  \item 
    We say some set $G\subseteq E$ is maximal $E$-$\Phi$-free if $G$ is
    $E$-$\Phi$-free and there is no  $E$-$\Phi$-free strict superset
    of $G$. 
  \item 
    For all $S\subseteq A$ we write $\lvert S\restriction{\Fam{F}} \rvert$ for
    \begin{equation*}
      \max\Set{\lvert F \rvert |  F\in \Fam{F}\text{\ with }F\subseteq S}.     
    \end{equation*}
    Set
    \begin{equation*}
      \nu^\Phi_\Fam{F}(E) = \min\Set{\lvert G\restriction{\Fam{F}} \rvert |
        G\text{\ maximal }E\text{-}\Phi\text{-free}
      }.
    \end{equation*}
    and for $n\in\N$, set 
    \begin{equation*}
      \nu^\Phi_\Fam{F}(n) = \inf \Set{\nu^\Phi_\Fam{F}(E) | E\in\Fam{F},
        \lvert E \rvert = n}\in \N\cup\{\infty\}
    \end{equation*}
    (where $\inf \emptyset = \infty$). 
  \end{enumerate}
\end{definition}
$\nu_{\Fam{F}}^\Phi$ measures the minimal growth rate of sets definable from
$\Phi$ with a finite set of parameters with respect to some infinite family
$\Fam{F}$. In most applications $\Fam{F}$ can be defined to be the set
of all subsets. In this case $\nu_{\Fam{F}}^\Phi$ just measures the
growth rate of sets definable from $\Phi$ with a finite set of
parameters. 
Let us comment on how such a function $\nu_{\Fam{F}}^\Phi$ is usually
used. Typical results on growth rate are of the form ``there are infinitely
many $n\in\N$ such that $\nu_{\Fam{F}}^\Phi(n) \leq p(n)$'' for a
certain polynomial $p$. If $\Fam{F}$ is the set of all subsets of the
domain of the given structures, this says that for infinitely many
values of $n$ there is a subset $E$ of size $n$ such that every maximal
$E$-$\Phi$ free set $G$ has size at most $p(n)$.

\section{Basic Results}
\label{sec:BasicResults}
In this Section we cite some results from \cite{HuKaSchl14} that turn
out to be useful in the following sections.

\begin{proposition}[Proposition~3.6 of \cite{HuKaSchl14}] \label{prop:PumpingLemma} Let $\alpha\geq 1$
  be an ordinal of countable cofinality and let $\Aut{A} = (S, \Sigma,
  I, F, \Delta)$ be an automaton with $\lvert S \rvert \leq m$.  For
  all $s_0,s_1\in S$ and $\sigma\in\Sigma$,
  \begin{equation*}
    s_0 \run{\sigma^{\omega^m}}{\Aut{A}} s_1 \Longleftrightarrow 
    s_0 \run{\sigma^{\omega^m\alpha}}{\Aut{A}} s_1.
  \end{equation*}
\end{proposition}

\begin{proposition}[cf.~Proposition~3.7 of  \cite{HuKaSchl14}] \label{prop:PumpingLemmaUnc} Let $\alpha\geq 1$ be an
  ordinal of uncountable cofinality and let $\Aut{A} = (S, \Sigma, I,
  F, \Delta)$ be an automaton with $\lvert S \rvert \leq m$.  For all
  $s_0,s_1\in S$ and $\sigma\in\Sigma$,
  \begin{equation*}
    s_0 \run{\sigma^{\omega_1}}{\Aut{A}} s_1 \Longleftrightarrow 
    s_0 \run{\sigma^{\alpha}}{\Aut{A}} s_1.
  \end{equation*}
\end{proposition}

\begin{lemma}[Lemma~3.19 of \cite{HuKaSchl14}]\label{lem:WellOrder}
  For every finite alphabet $\Sigma$, 
  there is an $\alpha$-automaton that recognises a well-order $\sqsubseteq$
  on the set $\FinWords{\Sigma}{\alpha}$. 
  Moreover the relation $\subseteq_\supp$ given by $w \subseteq_\supp
  v$ if and only if $\supp(w) \subseteq \supp(v)$ is
  \automatic{\alpha}. 
\end{lemma}

\section{The Growth Rate Technique}
\label{sec:GrowthRate}

\Delhomme proved the following bounds on the growth rates of maximal
$\Fam{F}$-$E$-$\Phi$-free sets in the word- and tree-automatic setting. 

\begin{proposition}[\cite{Delhomme04}]
  For each set $\Phi$ of word-automatic relations, 
  there is a constant $k$ such that
  $\nu_{\Fam{F}}^\Phi (n) \leq k\cdot n$ for infinitely many $n\in \N$. 

  For each set $\Phi$ of tree-automatic relations, 
  there is a constant $k$ such that
  $\nu_{\Fam{F}}^\Phi(n) \leq n^k$ for infinitely many $n\in \N$. 
\end{proposition}

The basic proof idea is to show that any $E$-$\Phi$-free set $G$ can be
transformed into an $E$-$\Phi$-free set $G'$ such that $\lvert G\rvert
= \lvert G' \rvert$ whose elements are all words (or trees) that have
a domain that is similar to the union of the domains of all parameters
from $E$. In order to prove a similar result we first provide a notion
of having similar domains for $(\alpha)$-words. 

\begin{definition} 
  Let $m\in\N$, $X$ a finite set of ordinals and $\beta$ and ordinal
  of the form
  \begin{equation*}
    \beta=\beta_{\sim m}+\omega^m n_m+\omega^{m-1} n_{m-1}+
    \dots + n_0.
  \end{equation*} 
  \begin{itemize}
  \item Let $U_m(\beta)$ denote the set of ordinals
    $\gamma=\gamma_{\sim m}+\omega^m l_m+\omega^{m-1} l_{m-1}+ \dots
    +l_0$ such that one of the following holds:
    \begin{itemize}
    \item $\gamma=\beta$,
    \item $\gamma_{\sim m}=\beta_{\sim m}$ and for $k$ maximal with
      $l_k\neq n_k$, we have $l_k\leq n_k+m$ and $l_i\leq m$ for all
      $i<k$, or
    \item $\gamma_{\sim m} = \beta_{\sim m} +\omega_1 c$ for some
      $1\leq c \leq m$ and $l_i \leq m$ for all $0\leq i \leq m$.
    \end{itemize}
  \item Let $U_m(X,\delta)= \left(\bigcup_{\gamma\in X\cup\{0,
        \delta\}} U_m(\gamma) \right) \cap \delta$.
  \item Let $U_m^1(X,\delta)=U_m(X,\delta)$ and
    $U_m^{i+1}(X,\delta)=U_m(U_m^i(X,\delta),\delta)$ for $i\in\N$.
  \end{itemize}
\end{definition}

A crucial observation is that, roughly speaking, 
there are few $(\alpha)$-words with support in 
$U^i_m(X, \alpha)$. Using the following abbreviations, we make this
idea precise in the following lemma:
\begin{enumerate}
\item $c_m(\beta)=\max_{i\leq m} n_i$,
\item $c_m(X)=\max_{\gamma\in X} c_m(\gamma)$, and
\item $d_m(X)=|\{\gamma_{\sim m}\mid \gamma\in X\cup\{0\}\}|$.
\end{enumerate}

\begin{lemma}
  \label{lem:BoundsOnU}
  Suppose that $X$ is a finite set of ordinals and $i\geq 1$. Then
  \begin{align*}
    &\lvert U_m^i(X,\alpha)\rvert \leq (c_m(X\cup\{\alpha\})+im)^{m+1}
    \cdot (i \cdot m+1) \cdot d_m(X\cup\{\alpha\}).
  \end{align*}
\end{lemma}

\begin{proof}
  A simple induction shows that all $\gamma \in U_m^i(\beta)$
  satisfy $\gamma_{\sim m} = \beta_{\sim m} + \omega_1 \cdot k$ for
  some $0\leq k \leq (i \cdot m)$.

  One also proves inductively that the coefficient of $\omega^j$ of an
  element of $U_m^i(X, \alpha)$ is bounded by $(c_m(w)+im)$.
\end{proof}

In this section, we fix an ordinal $\alpha$, and a finite set of
$(\alpha)$-automata 
\begin{equation*}
\Phi = (\Aut{A}_1, \Aut{A}_2, \dots, \Aut{A}_n).  
\end{equation*}
Without loss of generality, each automaton has the same state set $Q$.
We fix the constant $K = \lvert 2^{Q\times Q} \rvert^n +1$.

The following proposition contains the main technical result that
allows to use the relative growth technique for ordinal-automatic
structures. This proposition implies that for all $E$ and $\Phi$ there
is a $E$-$\Phi$-free set of maximal size with support in
$U_K(\supp(E))$. Since $U_K(\supp(E))$ is small this provides an upper bound
on the minimal size of maximal $E$-$\Phi$-free sets.

\begin{proposition}\label{prop:RelativeGrowthsOrdinalAutomatic}
  Let $E\subseteq \FinWords{\Sigma}{\alpha}$ and
  $v\in\FinWords{\Sigma}{\alpha}$. There is a word $w\in U_K(\supp(E),
  \alpha)$ such that $v \sim^\Phi_E w$.
\end{proposition}

For better readability we first provide a simple tool for the proof

\begin{lemma}
  Let $E\subseteq \FinWords{\Sigma}{\alpha}$ and
  $v\in\FinWords{\Sigma}{\alpha}$. Let $n\in\N$ and $\beta$ some
  ordinal such that $\beta+\omega^{n+1}\leq \alpha$.  If there is an
  ordinal $\gamma$ such that $\gamma < \gamma + \omega^n (K-1) \leq
  \beta < \gamma+\omega^{n+1}$ and
  \begin{equation}
    \label{eq:suppEMpty}
    \supp(E) \cap [\gamma, \gamma+\omega^{n+1}) = \emptyset    
  \end{equation}
  then there are natural numbers $n_1 <n_2 \leq K$ such that for
  \begin{equation*}
    w:= v \restriction{[0,\gamma+\omega^nn_1)} + 
    v \restriction{[\gamma+\omega^n n_2, \alpha)}
  \end{equation*}
  $w \sim^\Phi_E v$, i.e., for all $\bar e \in E^k$ and $\Aut{A}\in
  \Phi$
  \begin{equation*}
    \Aut{A} \text{\ accepts } v\otimes \bar e \text{\ iff }
    \Aut{A} \text{\ accepts } w\otimes \bar e 
  \end{equation*}
\end{lemma}
\begin{proof}
  Set $I=\Set{1, 2, \dots, n}$. We define the function
  \begin{equation*}
    f:\Set{0, 1, \dots, K} \to 2^{Q\times Q \times I}    
  \end{equation*}
  such that $f(j)$ contains $(q,p,i)$ if and only if there is a run of
  $\Aut{A}_i$ from state $q$ to state $p$ on $v\restriction{[\gamma,
    \gamma+\omega^n j)} \otimes \diamond^{\omega^n j}$. By choice of
  $K$ there are $n_1<n_2$ with $f(n_1) = f(n_2)$.  Thus,
  \begin{align*}
    &q \run{v\otimes \bar e}{\Aut{A}_i} p\\
    \xLeftrightarrow{\eqref{eq:suppEMpty}} &\exists{r,s} \left(q
      \run{(v\otimes \bar e)\restriction{[0, \gamma)}}{\Aut{A}_i} r
      \land r \run{v\restriction{[\gamma, \gamma+\omega^n n_2)}
        \otimes \diamond^{\omega^n n_2}}{\Aut{A}_i} s \land s
      \run{(v\otimes \bar e) \restriction{[\gamma+\omega^n n_2,
          \alpha)}}{\Aut{A}_i} p\right)\\
    \xLeftrightarrow[= f(n_2)]{f(n_1) } &\exists{r,s} \left(q
      \run{(v\otimes \bar e)\restriction{[0, \gamma)}}{\Aut{A}_i} r
      \land r \run{v\restriction{[\gamma, \gamma+\omega^n n_1)}
        \otimes \diamond^{\omega^n n_1}}{\Aut{A}_i} s \land s
      \run{(v\otimes \bar e) \restriction{[\gamma+\omega^n n_2,
          \alpha)}}{\Aut{A}_i} p\right)\\
    \xLeftrightarrow{\eqref{eq:suppEMpty}} &q \run{w\otimes \bar
      e}{\Aut{A}_i} p
  \end{align*}
  from which we immediately conclude that $v \sim^\Phi_E w$.  
\end{proof}

\begin{proof}
  [Proof of Proposition~\ref{prop:RelativeGrowthsOrdinalAutomatic}] The
  proof is by outer induction on $\lvert \supp(v) \setminus
  U_K(\supp(E), \alpha) \rvert$ and by inner (transfinite) induction
  on
  \begin{equation*}
    \beta = \min( \supp(v) \setminus U_K(\supp(E), \alpha)).
  \end{equation*}
  Fix the presentation
  \begin{equation*}
    \beta = \beta_{\sim K} + \omega^K b_K + \omega^{K-1} b_{K-1}
    +\dots+ b_0
  \end{equation*}
  with $b_0, \dots, b_K\in \N$
  and proceed as follows.
  \begin{itemize}
  \item If there is some $n\leq K$ such that $b_n+1-K \geq 0$ and
    \begin{equation*}
      \begin{aligned}
        &(\supp(E)\cup\{\alpha\}) \cap [\epsilon_1, \epsilon_2) =
        \emptyset \text{\ where }\\
        &\epsilon_1 = \beta_{\sim K} + \omega^K b_K + \omega^{K-1}
        b_{K-1} +\dots+
        \omega^n (b_n+1-K) \text{\ and}\\
        &\epsilon_2 = \beta_{\sim K} + \omega^K b_K + \omega^{K-1}
        b_{K-1} +\dots+ \omega^{n+1} (b_{n+1}+1),
      \end{aligned}
    \end{equation*}
    then we can apply the previous lemma and obtain a word $v'$ such
    that $v \sim^\Phi_E v'$ and 
    \begin{equation*}
      \lvert \supp(v') \setminus
      U_K(\supp(E), \alpha) \rvert  <   
      \lvert \supp(v) \setminus
      U_K(\supp(E), \alpha) \rvert
    \end{equation*} 
    or
    \begin{equation*}
      \beta' = \min( \supp(v') \setminus U_K(\supp(E), \alpha)).
    \end{equation*}
    has a presentation
    \begin{equation*}
      \beta' = \beta_{\sim K} + \omega^K b_K + \omega^{K-1} b_{K-1}
      +\dots+ \omega^{n+1} b_{n+1} + \omega^n b' + \omega^{n-1}
      b_{n-1} + \dots + b_0
    \end{equation*}
    with $b' < b_n$.
  \item Assume that the conditions for the previous case are not
    satisfied. Either $b_i \leq K-1 $ for $0\leq i \leq K$ or
    there is a minimal $i\leq K$ such that $b_i \geq K$. 
    We first show that the latter case cannot occur. 
    Assuming $b_i \geq K$ we have
    \begin{equation*}
      \begin{aligned}
        &(\supp(E)\cup\{\alpha\}) \cap
        [\epsilon_1 , \epsilon_2) \neq \emptyset \text{\ where}\\
        &\epsilon_1=\beta_{\sim K} + \omega^K b_K + \omega^{K-1}
        b_{K-1}
        +\dots+ \omega^i (b_i + 1 - K) \text{\ and}\\
        &\epsilon_2= \beta_{\sim K} + \omega^K b_K + \omega^{K-1}
        b_{K-1} +\dots+ \omega^{i+1} (b_{i+1}+1).
      \end{aligned}
    \end{equation*}
    Thus, there is some $\gamma \in \supp(E)\cup\{\alpha\}$ with
    \begin{equation*}
      \gamma = \beta_{\sim K} + \omega^K b_K + \omega^{K-1} b_{K-1}
      +\dots+ \omega^{i+1} b_{i+1} + \omega^i c_i + \omega^{i-1} c_{i-1}
      + \dots + c_0
    \end{equation*}
    such that $c_i + K-1 > b_i$ whence $\beta \in U_K(\gamma)\subseteq
    U_K(\supp(E), \alpha)$ contradicting the definition of $\beta$.

    Thus, we can assume that $b_i \leq K-1$ for all $0\leq i\leq K$. By
    definition of $\beta$ we conclude that $\beta_{\sim K} \neq
    \gamma_{\sim K} + \omega_1 c$ for $c\in\{0, 1, \dots, K\}$ for all
    $\gamma\in \supp(E)\cup \{0,\alpha\}$.  We proceed with one of the
    following cases depending on the cofinality of $\beta_{\sim K}$.
    \begin{enumerate}
    \item If $\beta_{\sim K}$ has countable cofinality, let
      \begin{align*}
        &\gamma = \max( ( \supp(E \cup\{v\} ) \cap \beta_{\sim K}
        )+1 \\
        &\delta = \min((\supp(E) \cap [\beta, \alpha)) \cup
        \{\alpha\})
        \text{\ and}\\
        &\delta' = \max( \supp(v) \cap \delta_{\sim K})+1.
      \end{align*}
      From the definition of $\beta$ it follows that $\gamma_{\sim K}
      < \beta_{\sim K} < \delta_{\sim K}$ and that $\sup(E) \cap
      [\gamma, \delta) = \emptyset$.  Note that $[\gamma, \beta_{\sim
        K})$ is of shape $\omega^{K+1} \eta_1$ for some ordinal
      $\eta_1\geq 1$ of countable cofinality.  By definition of
      $\delta'$, $[\delta', \delta_{\sim K})$ is of shape
      $\omega^{K+1} \eta_2$ for some ordinal $\eta_2\geq 1$.  Choose
      an ordinal $\eta$ such that $[\beta_{\sim K}, \delta') + \eta$
      is isomorphic to $[\gamma, \delta_{\sim K})$ and define
      \begin{equation*}
        w:= v\restriction{[0, \gamma)} + \diamond^{\omega^{K}} +
        v\restriction{[\beta_{\sim K}, \delta')} + \diamond^{\eta} +
        v\restriction{[\delta_{\sim K}, \alpha)}.
      \end{equation*}
      For all $\bar e \in E^k$ and $\Aut{A}\in A$ we conclude that
      \begin{align*}
        &q \run{v \otimes \bar e}{\Aut{A}} p \\ \Leftrightarrow&
        \exists{r,s,t,u \in Q} \left(
          \begin{alignedat}{3}
            &q \run{(v \otimes \bar
              e)\restriction{[0,\gamma)}}{\Aut{A}} r& &\land r
            \run{\diamond^{\omega^{K+1} \eta_1}}{\Aut{A}} s& &\land s
            \run{(v \otimes \bar e)\restriction{[\beta_{\sim K},
                \delta')}}{\Aut{A}} t \\
            && &\land t \run{\diamond^{\omega^{K+1}} \eta_2}{\Aut{A}}
            u & &\land u \run{(v \otimes \bar
              e)\restriction{[\delta_{\sim K}, \alpha)}}{\Aut{A}} p
          \end{alignedat}
        \right)
        \\
        \Leftrightarrow& \exists{r,s,t,u \in Q} \left(
          \begin{alignedat}{3}
            q \run{(v \otimes \bar
              e)\restriction{[0,\gamma)}}{\Aut{A}} r &\land r
            \run{\diamond^{\omega^{K}}}{\Aut{A}} s &&\land s
            \run{(v\otimes \bar e)\restriction{[\beta_{\sim
                  K},\delta')}}{\Aut{A}} t\\
            &\land t \run{\diamond^{\eta}}{\Aut{A}} u &&\land u
            \run{(v \otimes \bar e)\restriction{[\delta_{\sim K},
                \alpha)}}{\Aut{A}} p
          \end{alignedat}
        \right)  \\
        \Leftrightarrow & q \run{w\otimes \bar e}{\Aut{A}} p
      \end{align*}
      whence $w \sim^\Phi_E v$. Moreover
      \begin{equation*}
        \lvert \supp(w) \setminus U_K(\supp(E), \alpha) \rvert <
        \lvert \supp(v) \setminus U_K(\supp(E), \alpha) \rvert    
      \end{equation*}
      because the letter at position $\beta$ in $v$ has been shifted
      to position
      \begin{equation*}
        \gamma_{\sim K} + \omega^K (b_K+1) + \omega^{K-1}
        b_{k-1} + \omega^{K-2} b_{k-2} + \dots + b_0.
      \end{equation*}
      Since all $b_i \leq K-1$ we conclude that this  position belongs to
      $U_K(\supp(E), \alpha)$.
    \item If $\beta_{\sim K}$ has uncountable cofinality, Let
      \begin{align*}
        &\gamma = \max( \supp(E) \cap \beta_{\sim K}
        )+1 \\
        &\delta = \min((\supp(E) \cap [\beta, \alpha)) \cup
        \{\alpha\})
        \text{\ and}\\
        &\delta' = \max( \supp(v) \cap \delta_{\sim K})+1.
      \end{align*}
      From the definition of $\beta$ and the uncountable cofinality of
      $\beta_{\sim K}$ it follows that 
      $\gamma_{\sim K} +\omega_1 K < \beta_{\sim K} < \delta_{\sim K}$
      and that 
      $\sup(E) \cap [\gamma, \delta) = \emptyset$.  Let
      \begin{equation*}
        \gamma' = \max( \supp(E\cup\{v\}) \cap \beta_{\sim K})+1
      \end{equation*}
      and note that $\gamma' \leq \gamma_{\sim K} + \omega_1 (K+1)$
      because $\gamma'\in U_K(E, \alpha)$ since $\beta$ has
      been chosen minimal.  
      If $\gamma' \geq \gamma_{\sim K} + \omega_1 K$, then we do the
      following 
      preparatory step that locally changes $v$ to some $v'$ such that
      $ v \sim^\Phi_E v'$ and shrinking the corresponding value of
      $\gamma'$.
      For this purpose, note that 
      $e\restriction{[\gamma, \beta)} =\diamond^{[\gamma, \beta)}$ for
      all $e\in E$.  
      By choice of $K$ there are numbers $i < j \leq K$ such
      that for all $\Aut{A}\in\Phi$, all $q,p\in Q$ and all $\bar e
      \in E^k$ we have
      \begin{equation}\label{eq:omega1ijRunsEquality}
        q \run{(v\otimes \bar e)\restriction{[\gamma, \gamma +
            \omega_1 i)}}{\Aut{A}} p \Longleftrightarrow 
        q \run{(v\otimes \bar e)\restriction{[\gamma, \gamma +
            \omega_1 j)}}{\Aut{A}} p
      \end{equation}
      Choose an ordinal $\eta$ such that $\omega_1 \cdot (j-i) +
      [\gamma+\omega_1 j, \gamma') = [\gamma+\omega_1 j, \gamma') +
      \eta$ and set
      \begin{align*}
        v' = v\restriction{[0, \gamma_{\sim K} + \omega_1 i)} +
        v\restriction{[\gamma_{\sim K} + \omega_1 j, \gamma')} +
        \diamond^\eta + v\restriction{[\gamma', \alpha)}.
      \end{align*}
      Now $v \sim^\Phi_E v'$ because for all $\bar e \in E^k$
      \begin{align*}
        &q \run{v \otimes \bar e}{\Aut{A}} p \\
        \Leftrightarrow& \exists{r,s,t\in Q} \left(
          \begin{alignedat}{3}
            &q \run{(v \otimes \bar e)\restriction{[0,\gamma +
                \omega_1 j)}}{\Aut{A}} r 
            &&\land r \run{(v
              \otimes \bar e)\restriction{[\gamma + \omega_1 j,
                \gamma')}}{\Aut{A}} s \\
            \land\ &s \run{(v \otimes \bar e)\restriction{[\gamma',
                \beta_{\sim K})}}{\Aut{A}} t &&\land t \run{(v \otimes
              \bar e)\restriction{[\beta_{\sim K}, \alpha)}}{\Aut{A}}
            p
          \end{alignedat}
        \right)\\
        \xLeftrightarrow{Eq.~\eqref{eq:omega1ijRunsEquality}}& \exists{r,s,t\in Q} \left(
          \begin{alignedat}{3}
            &q \run{(v \otimes \bar e)\restriction{[0,\gamma +
                \omega_1 i)}}{\Aut{A}} r 
            &&\land r \run{(v \otimes \bar e)\restriction{[\gamma + \omega_1 j,
                \gamma')}}{\Aut{A}} s \\
            \land\ &s \run{\diamond^{[\gamma', \beta_{\sim
                  K})}}{\Aut{A}} t &&\land t \run{(v \otimes \bar
              e)\restriction{[\beta_{\sim K}, \alpha)}}{\Aut{A}} p
          \end{alignedat}
        \right)\\
        \xLeftrightarrow{Prop.~\ref{prop:PumpingLemmaUnc}}&
          \exists{r,s,t\in Q} \left( 
          \begin{alignedat}{3}
            &q \run{(v' \otimes \bar e)\restriction{[0,\gamma + \omega_1 i)}}{\Aut{A}} r 
            &&\land 
            r \run{(v \otimes \bar e)\restriction{[\gamma + \omega_1 j, \gamma')}}{\Aut{A}} s \\
            \land\ &s \run{\diamond^{\eta+[\gamma', \beta_{\sim
                  K})}}{\Aut{A}} t &&\land t \run{(v \otimes \bar
              e)\restriction{[\beta_{\sim K}, \alpha)}}{\Aut{A}} p
          \end{alignedat}
        \right)\\
        \Leftrightarrow & q \run{v'\otimes \bar e}{\Aut{A}} p
      \end{align*}
      Note that the definitions of $\beta$, $\gamma$, $\delta$ and
      $\delta'$ with respect to $v'$ agree with those for $v$. Thus,
      from now on we replace $v$ by $v'$ whence we can assume that
      $\gamma' < \gamma + \omega_1 K$.

      Note that $[\gamma', \beta_{\sim K})$ is
      of shape $\omega^{K+1} \eta_1$ for some ordinal $\eta_1\geq 1$
      of uncountable cofinality.  By definition of $\delta'$,
      $[\delta', \delta_{\sim K})$ is of shape $\omega^{K+1} \eta_2$
      for some ordinal $\eta_2\geq 1$.  Choose an ordinal $\eta$ such
      that $[\beta_{\sim K}, \delta') + \eta$ is isomorphic to
      $[\gamma' +\omega_1, \delta_{\sim K})$ and define
      \begin{equation*}
        w:= v\restriction{[0, \gamma')} + \diamond^{\omega_1} +
        v\restriction{[\beta_{\sim K}, \delta')} + \diamond^{\eta} +
        v\restriction{[\delta_{\sim K}, \alpha)}.
      \end{equation*}
      For all $\bar e \in E^k$ and $\Aut{A}\in A$ we conclude that
      \begin{align*}
        &q \run{v \otimes \bar e}{\Aut{A}} p \\
        \Leftrightarrow& \exists{r,s,t,u\in Q} \left(
          \begin{alignedat}{3}
            q \run{(v \otimes \bar
              e)\restriction{[0,\gamma')}}{\Aut{A}} r &\land r
            \run{\diamond^{\omega^{K+1} \eta_1}}{\Aut{A}} s &&\land s
            \run{(v \otimes \bar e)\restriction{[\beta_{\sim K},
                \delta')}}{\Aut{A}} t \\
            &\land t \run{\diamond^{\omega^{K+1}} \eta_2}{\Aut{A}} u
            &&\land u \run{(v \otimes \bar
              e)\restriction{[\delta_{\sim K}, \alpha)}}{\Aut{A}} p
          \end{alignedat}
        \right)\\
        \Leftrightarrow& \exists{r,s,t,u\in Q} \left(
          \begin{alignedat}{2}
            q \run{(v \otimes \bar
              e)\restriction{[0,\gamma')}}{\Aut{A}} r &\land r
            \run{\diamond^{\omega_1}}{\Aut{A}} s &&\land
            s \run{(v \otimes \bar e)\restriction{[\beta_{\sim K}, \delta')}}{\Aut{A}} t\\
            &\land t \run{\diamond^{\eta}}{\Aut{A}} u &&\land u
            \run{(v \otimes \bar e)\restriction{[\delta_{\sim K},
                \alpha)}}{\Aut{A}} p
          \end{alignedat}
        \right) \\
        \Leftrightarrow & q \run{w\otimes \bar e}{\Aut{A}} p
      \end{align*}
      whence $w \sim^\Phi_E v$. Moreover
      \begin{equation*}
        \lvert \supp(w) \setminus U_K(\supp(E), \alpha) \rvert <
        \lvert \supp(v) \setminus U_K(\supp(E), \alpha) \rvert    
      \end{equation*}
      because the letter at position $\beta$ in $v$ has been shifted
      to position
      \begin{equation*}
        \gamma'_{\sim K}+\omega_1  + \omega^K b_K + \omega^{K-1}
        b_{k-1} + \omega^{K-2} b_{k-2} + \dots + b_0
      \end{equation*}
      and since all $b_i<K-1$ we conclude that this is position is
      contained in $U_K(\supp(E), \alpha)$ because $\gamma_{\sim K} =
      \eta_{\sim K}+ \omega_1 c$ for some 
      $\eta\in \supp(E)\cup\{0\}$ and some $c\in\{0, 1, \dots, K-1\}$.
      \qedhere
    \end{enumerate}
  \end{itemize}
\end{proof}

The previous result allows us to directly deduce the following bound
on the growth rates of ordinal-automatic relations. 

\begin{theorem}\label{thm:GrowthRageOrdinalAutomatic}
  Fix an infinite family $\Fam{F}$ of sets of $(\alpha)$-words with
  $\emptyset\in\Fam{F}$. 
  For every $c > 1$, $\nu^\Phi(n) \leq n^c$ for infinitely
  many $n\in \N$. 
\end{theorem}
\begin{proof}
  Heading for a contradiction, assume that there is a natural number
  $N$ such that
  \begin{equation*}
    \forall{E\in\Fam{F} \text{\ with }\lvert E \rvert >N} \forall{G
      \text{\ maximal $E$-$\Phi$-free}} \exists{F\in\Fam{F}}
      (F\subseteq G \text{ and } \lvert F \rvert \geq \lvert E \rvert^c)
  \end{equation*}
  where $c > 1$.
  
  Take a finite set $F_0$ of parameters with $\lvert F_0 \rvert > N$.  Having
  defined a finite set $F_{i-1}$ such that
  $\supp(F_{i-1})\in U_K^{i-1}(\supp(F_0))$ we can use the previous
  lemma to choose some maximal $F_{i-1}$-$\Phi$-free set $G_i$ with
  $\supp(G_i) \in U_K^{i}(\supp(F_0))$.  By assumption, there is some
  $F_i\in\Fam{F}$ with $F_i\subseteq G_i$ and  
  $\lvert F_i \rvert \geq \lvert F_{i-1} \rvert^c$.
  By induction we obtain $\lvert F_{i} \rvert \geq \lvert F_0
  \rvert^{c^i}$.

  On the other hand,  all elements of $F_i$ have support in
  $U_K^{i}(\supp(F_0))$ which by  Lemma~\ref{lem:BoundsOnU} is at most
  of size 
  \begin{equation*}
    (\lvert \Sigma \rvert +1)^{c_0 (iK+1) ( c_1 + K i)^K}
  \end{equation*}
  for some constants $c_0$ and $c_1$.  Since $c^i$ grows faster than
  any polynomial in $i$ we have 
  $c^i > c_0 (i+1) ( c_1 + K i)^K$ 
  for some large $i$
  which leads to a contradiction.  
\end{proof}

\section{Applications}
\label{sec:Applications}

\subsection{Random Graph}
\label{sec:RandomGraph}

The \emph{random graph (or Rado graph)} $(V, E)$ is the unique countable graph that has the
property that for any choice of finite subsets $V_0, V_1\subseteq V$
there is a node $v'$ which is adjacent to every element of $V_0$ but
not adjacent to any element of $V_1$.

\begin{theorem}
  Given an ordinal $\alpha$, 
  the random graph is  not
  \automatic{\alpha}.
\end{theorem}
\begin{proof}
  Heading for a contradiction assume that the random graph was
  \automatic{\alpha}.  As shown by \Delhomme\cite{Delhomme04}, the
  random graph satisfies $\nu_{\Fam{F}}^\Phi(n) = 2^n$ for all
  $n\in\N$ where $\Phi$ consists of only one automaton recognising the
  edge relation of the random graph and $\Fam{F}$ contains all subsets
  of the domain of the random graph.  This contradicts
  Proposition~\ref{prop:RelativeGrowthsOrdinalAutomatic}.
\end{proof}
\begin{remark}
  Similarly, taking $\Fam{F}$ to be the family of all antichains, one
  proves that the random partial order is not \automatic{\alpha}
  (cf.~\cite{KhoussainovNRS07} for an analogous result for automatic structures). 
\end{remark}

\subsection{Integral Domains}
\label{sec:IntegralDom}

In this part we show that there is no infinite \automatic{\alpha}
integral domain for any $\alpha < \OmegaPlus$. 
We cannot use the growth rate theorem directly but use a variant of
its proof. 
The difference is that we do not use a fixed set of relations $\Phi$
when defining the sequence $(F_i)_{i\in \N}$ but in each step we take a
different relation but ensure that we can still apply 
Proposition~\ref{prop:RelativeGrowthsOrdinalAutomatic} with a fixed
constant $K$ in each step. This is ensured by using  relations defined by a
fixed automaton $\Aut{A}$ which has an additional parameter which is
chosen very carefully. 
In fact, we follow the proof of
Khoussainov et al. \cite{KhoussainovNRS07} from the automatic
case. Let us recall the basic definitions and some observations 
from their proof. An \emph{integral domain} is a commutative ring with
identity $(D,+, \cdot, 0, 1)$ such that 
$d \cdot e = 0 \Rightarrow d=0$ or $e=0$  for all $d,e\in D$.

\begin{lemma}[cf.~Proof of Theorem~3.10 from \cite{KhoussainovNRS07}]
  Let $(D, +, \cdot, 0, 1)$ be an integral domain and $E\subseteq D$ a
  finite subset. There is some $d\in D$ such that for all $a_1, a_2,
  b_1, b_2 \in E$,  if $a_1 d + b_1 = a_2 d + b_2$, then
  $a_1 = a_2$ and $b_1 = b_2$, i.e., the function $f_d:E^2 \to D$,
  $(e_1,e_2)\mapsto e_1 d+ e_2$ is injective. 
\end{lemma}

\begin{proposition}
  Let $\Struc{A} = (D, +, \cdot, 0, 1)$ be an \automatic{\alpha}
  integral domain 
  for some \mbox{$\alpha < \OmegaPlus$}.\footnote{Without loss of
    generality, we can assume that $\Struc{A}$ has a injective representation by
    the automata $(\Aut{A}_D, \Aut{A}_+, \Aut{A}_\cdot, \Aut{A}_0,
    \Aut{A}_1)$  such that $L(\Aut{A}_D) = D$, i.e., $D$ is a set of
    $(\alpha)$-words (cf.~\cite{HuKaSchl14}).}
  
  There is a constant $m$ such that for every
  finite set $X\subseteq \alpha$ of ordinals 
  we have 
  \begin{equation*}
    \lvert\Set{d\in D\mid \supp(d)\subseteq U_m(X, \alpha)}\rvert \geq
    \lvert\Set{d\in D\mid \supp(d)\subseteq X}\rvert^2. 
  \end{equation*}
\end{proposition}
\begin{proof}
  As an abbreviation, we use the expression $x_1, \dots, x_n
  \subseteq_\supp y$ for 
  \begin{equation*}
    x_1 \subseteq_\supp y \land \dots \land
    x_n\subseteq_\supp y.    
  \end{equation*}
  Let $\psi(x, p)$ denote the formula
  \begin{equation*}
    x\in D \land
    \forall{a_1,a_2,b_1,b_2\subseteq_\supp p}
    \left( ( a_1 x+b_1 = a_2 x+b_2) \rightarrow
    (a_1 = a_2 \land b_1 = b_2) \right)
  \end{equation*}
  and $\psi_{\min}(x,p)$ the formula
  \begin{equation*}
    \psi(x,p) \land \forall y (\psi(y,p) \rightarrow x \sqsubseteq y)
  \end{equation*}
  where $\sqsubseteq$ denotes the \automatic{\alpha} well-order from
  Lemma ~\ref{lem:WellOrder}.
  Due to the previous lemma for every
  $p\in\FinWords{\Sigma}{\alpha}$ there is a unique $x$ satisfying
  $\psi_{\min}(x,p)$. Moreover, the map \mbox{$f:(a,b)\mapsto ax+b$} is
  injective when the domain is restricted to words with support
  contained in $\supp(p)$. 

  Since $\sqsubseteq$ is \automatic{\alpha} 
  and \automatic{\alpha}
  structures are close under first-order definitions, 
  there is an automaton $\Aut{A}_\phi$ corresponding to the following
  formula
  \begin{equation*}
    \phi(p,a,b,c) = a\subseteq_\supp p \land b\subseteq_\supp p
    \land \exists x \left(\psi_{\min}(x,P) \land  c= ax+b\right).    
  \end{equation*}
  For each finite set $X$ of ordinals, choose an $(\alpha)$-word $p_X$
  such that $\supp(p) = X$. Set
  \begin{align*}
    D_X &= \Set{d\in D | c\subseteq_\supp p}
    \text{\ and}\\
    F_X &= \Set{c \in D |
      \exists{a,b\in\FinWords{\Sigma}{\alpha}}
      \Aut{A}_\phi \text{\ accepts } (p_X, a, b, c)}.
  \end{align*}
  Since we are dealing with an injective presentation, 
  Proposition~\ref{prop:RelativeGrowthsOrdinalAutomatic} 
  implies that for every $a,b\in D_X$ there is some $c_{a,b}\in F_X$ such
  that $\Aut{A}_\phi$ accepts $(p_X, a, b, c_{a,b})$ and
  $\supp(c_{a,b}) \subseteq U_K(\supp(p_X)\cup \supp(a) \cup \supp(c),
  \alpha) = U_K(X, \alpha)$. Moreover, $c_{a,b}=c_{a',b'}$ implies
  $a=a'$ and $b=b'$ whence we conclude that 
  \begin{equation*}
    \lvert \Set{d\in D | \supp(d) \subseteq U_K(X, \alpha)}\rvert \geq
    \lvert F_X \rvert \geq \lvert D_X\rvert^2.    
  \end{equation*}
\end{proof}
\begin{remark}
  We crucially rely on $\alpha < \OmegaPlus$ because otherwise we
  cannot be sure that there is an automaton corresponding to
  $\psi(x,p)$. 
\end{remark}

\begin{corollary}
  Let $\alpha < \OmegaPlus$. There is no \automatic{\alpha} infinite
  integral domain. In particular, there is no \automatic{\alpha}
  infinite field. 
\end{corollary}
\begin{proof}
  Assume $D$ is the domain of an \automatic{\alpha} infinite integral
  domain. Choose two elements $d_1\neq d_2$ from $D$ and let 
  $X = \supp(d_1)\cup \supp(d_2)$.  Set 
  \begin{equation*}
    F_0 = \Set{d\in D | \supp(d)\subseteq \supp(d_1)\cup\supp(d_2)}.    
  \end{equation*}
  Iterated
  application of the previous lemma yields that
  \begin{equation*}
    F_i:= \Set{d\in D | \supp(d)\subseteq U^i_K(\supp(F_0),\alpha)}    
  \end{equation*}
  satisfies
  $\lvert F_{i+1} \rvert \geq \lvert F_i \rvert^2$. 
  Since $\lvert F_0 \rvert \geq 2$ we conclude that
  $\lvert F_n\rvert \geq 2^{2^n}$ and 
  $\supp(F_n)\subseteq U^n_K(X, \alpha)$.
  From Lemma~\ref{lem:BoundsOnU} we conclude that there are only
  $2^{p(n)}$ many elements in $D$ with support in $U^n_K(X, \alpha)$
  for some polynomial $p(n)$ which results in a contradiction for
  large $n$. 
\end{proof}

\section{Optimality of the Bound on the Growth-Rate}

Recall that word-automatic structures satisfy that
$\nu_{\Fam{F}}^{\Phi}(n) < n \cdot k$ for some constant $k$ where
$\Fam{F}$ is a family as before and $\Phi$ is a finite set of
word-automatic relations. In contrast, our bound for
\automatic{\alpha} structures is only $n^k$ for every constant $k > 1$. 
In this section, we give an example that, if $\alpha\geq \omega^2$,
then there are \automatic{\alpha}
structures violating any bound of the form $n \cdot k$ for every
constant $k$.\footnote{It is easily shown that for $\alpha<\omega^2$ every
  \automatic{\alpha} structure is word-automatic and vice versa whence
  the stronger bound from the word-automatic case applies.}
\begin{definition}
  For every $n\in\N$, let 
  \begin{equation*}
    D_n = \Set{\omega n_1 + n_2 | n_1 + n_2 \leq n}
  \end{equation*}
  and let $T_n$ be the set $(\omega^2)$-words over
  $\Sigma=\Set{a,b,\diamond}$ such that 
  $w\in T_n$ if and only if $\supp(w) = D_n$. 
  For 
  $i\in\{a,b\}$, we also define functions 
  \begin{align*}
    f_i:\FinWords{\Sigma}{\omega^2}  \times
  \FinWords{\Sigma}{\omega^2} \to \FinWords{\Sigma}{\omega^2} \text{ by}\\
  f_i(w,v)(\alpha) =
  \begin{cases}
    i &\text{if }\alpha=0,\\
    w(\omega n_1 + n_2) &\text{if }\alpha = \omega n_1 + n_2+1 \\
    v(\omega n_1) &\text{if } \alpha = \omega (n_1+1). 
  \end{cases}
  \end{align*}
\end{definition}

It is not difficult to see that the graphs of $f_a$ and $f_b$ are
\automatic{\omega^2} relations. Let $\Phi$ consist of two automata, one
corresponding to the graph of $f_a$ and one corresponding to the graph
of $f_b$. 
It is straightforward to verify that 
$T_{n+1} = f_a(T_n\times T_n) \cup f_b(T_n \times T_n)$. 
Moreover, since $f_a$ and $f_b$ are functions, it follows that
any
maximal $T_n$-$\Phi$-free set is of the form $T_{n+1}\cup \{w\}$ for
some $(\omega^2)$-word $w\notin T_{n+1}$.
A simple calculation shows that $\lvert D_n \rvert =
\frac{(n+1)(n+2)}{2}$ whence 
$\lvert T_n \rvert = 2^{\frac{(n+1)(n+2)}{2}}$. 

\begin{corollary}
  Setting $\Fam{F}$ to be the family of the sets $T_n$ for every
  $n\in\N$, we obtain that
  $\nu_{\Fam{F}}^\Phi(m) =
  \begin{cases}
    m \cdot 2^{n+2} &\text{if } m = 2^{\frac{(n+1)(n+2)}{2}}, \\
    \infty &\text{otherwise.}
  \end{cases}
  $
\end{corollary}
\begin{proof}
  Just note that $\lvert T_{n+1} \rvert = 2^{\frac{(n+3)(n+2)}{2}} =
    2^{\frac{(n+1)(n+2)}{2}+(n+2)} = \lvert T_n \rvert \cdot
    2^{n+2}$. 
\end{proof}
Since there for every constant $k$ there is some value $n_0\in\N$ such
that $2^{n+2}\geq k$ this shows that 
$\nu_{\Fam{F}}^\Phi(m) \leq  m \cdot k$ only holds for finitely many
$m\in\N$. This shows that there is no word-automatic presentation of
the \automatic{\omega^2}  structure 
$(\FinWords{\Sigma}{\omega^2}, f_a, f_b)$ and that the bound in
Theorem~\ref{thm:GrowthRageOrdinalAutomatic} cannot be replaced by
$n\cdot k$. 

\bibliography{bib}
\bibliographystyle{jalc}

\end{document}